\documentclass[a4paper,USenglish,10pt]{article}
\usepackage{latexsym}
\usepackage{amsmath,amssymb,mathtools}
\usepackage{amsfonts}
\usepackage{cite}
\usepackage{hyperref}
\usepackage[linesnumbered,lined,ruled,noend,vlined]{algorithm2e}
\usepackage{apxproof}
\usepackage{cleveref}
\usepackage{tabularx,environ}
\usepackage{comment}
\usepackage{url}
\usepackage{authblk}
\usepackage[full]{complexity}
\usepackage{lscape}
\usepackage{xifthen}
\usepackage{lineno}
\usepackage{siunitx}
\usepackage{color}
\usepackage{tikz}
\usetikzlibrary{decorations.pathmorphing}

\newcommand{\revise}[1]{#1}

\newcommand{\set}[1]{\{#1\}}
\newcommand{\relmiddle}[1]{\mathrel{}\middle#1\mathrel{}}
\newcommand{\inset}[2]{\left\{#1 \relmiddle| #2\right\}}
\newcommand{\size}[1]{| #1|}
\newcommand{\order}[1]{O(#1)}

\newtheoremrep{theorem}{Theorem}
\newtheoremrep{lemma}[theorem]{Lemma}
\newtheoremrep{proposition}[theorem]{Proposition}
\newtheoremrep{corollary}[theorem]{Corollary}

\makeatletter
\@ifpackageloaded{algorithm2e}{
\SetKwProg{Fn}{Function}{}{}
\SetKwProg{Procedure}{Procedure}{}{}
\SetKwProg{Subprocedure}{Subprocedure}{}{}
\SetKwComment{tcc}{//}{}
\SetKwFunction{Output}{Output}%
\SetKwFunction{MTC}{Min-CEDS}
\SetKwFunction{best}{Min-CEDS-Best}
\SetKwInOut{AlgInput}{Input}
\SetKwInOut{AlgOutput}{Output}
\SetKwInOut{AlgPrecondition}{Pre-conditions}
\SetKwInOut{AlgInvariant}{Invariants}
}{}
\makeatother

\makeatletter
\newcommand{\problemtitle}[1]{\gdef\@problemtitle{#1}}
\newcommand{\probleminput}[1]{\gdef\@probleminput{#1}}
\newcommand{\problemoutput}[1]{\gdef\@problemoutput{#1}}
\NewEnviron{myproblem}{
  \problemtitle{}\probleminput{}\problemoutput{}
  \BODY
  \par\addvspace{.5\baselineskip}
  \noindent{
    \framebox[0.9\textwidth][c]{
        \begin{tabularx}{0.9\textwidth}{@{\hspace{\parindent}} l X c}
            \multicolumn{2}{@{\hspace{\parindent}}l}{{\sc \@problemtitle}} \\
            \textbf{Input:}  & \@probleminput \\
            \textbf{Output:} & \@problemoutput
        \end{tabularx}
        }
      }
  \par\addvspace{.5\baselineskip}
}
\makeatother

\newenvironment{claim}[1]{\par\noindent\underline{Claim:}\space#1}{}
\newcommand{\claimqed}{\hfill $\triangleleft$}
\newenvironment{claimproof}[1]{\par\noindent \textit{Proof of Claim. }\space#1}{\claimqed\\}

\newcommand{\comp}[1]{\mu(#1)}

\newcommand{\wpp}[1]{W(#1)}
\newcommand{\distor}[2]{\delta(#1, #2)}

\newcommand{\OPT}{\overline{\rm OPT}}

\usepackage{xifthen}

\newcommand{\neigh}[2][]{\ifthenelse{\isempty{#1}}{\mathtt{neigh}(#2)}{{\tt neigh}_{\text{\rom{#1}}}(#2)}}

\makeatletter
\newcommand*{\rom}[1]{\expandafter\@slowromancap\romannumeral #1@}
\makeatother

\begin{document}

\title{An Approximation Algorithm for \texorpdfstring{$K$}{K}-best Enumeration of Minimal Connected Edge Dominating Sets with Cardinality Constraints}
\author[1]{Kazuhiro Kurita}
\author[2]{Kunihiro Wasa}


\affil[1]{Nagoya University, Nagoya, Japan, {kurita@i.nagoya-u.ac.jp}}
\affil[2]{Hosei University of Technology, Tokyo, Japan, {wasa@hosei.ac.jp}}

\maketitle

\begin{abstract}
    \emph{$K$-best enumeration}, which asks to output $k$-best solutions without duplication, is a helpful tool in data analysis for many fields.
    In such fields, graphs typically represent data. Thus subgraph enumeration has been paid much attention to such fields. 
    However, $k$-best enumeration tends to be intractable since, in many cases, finding one optimum solution is \NP-hard. 
    To overcome this difficulty, 
    we combine $k$-best enumeration with a concept of enumeration algorithms called \emph{approximation enumeration algorithms}.
    As a main result, 
    we propose a $4$-approximation algorithm for minimal connected edge dominating sets which outputs $k$ minimal solutions with cardinality at most $4\cdot\OPT$, 
    where $\OPT$ is the cardinality of a minimum solution which is \emph{not} outputted by the algorithm.
    Our proposed algorithm runs in $\order{nm^2\Delta}$ delay, where $n$, $m$, $\Delta$ are the number of vertices, the number of edges, and the maximum degree of an input graph. 
\end{abstract}

\section{Introduction}
Enumeration and finding multiple solutions are required tasks in various fields. 
Such tasks have been used in operations research~\cite{doi:10.1287/ijoc.2020.1028}, machine learning~\cite{DBLP:conf/aaai/HaraI18,DBLP:conf/aaai/HanakaKKO21,DBLP:conf/ijcai/BasteFJMOPR20}, data mining~\cite{DBLP:conf/fimi/UnoKA04,DBLP:conf/kdd/ConteMSGMV18}, and database theory~\cite{DBLP:journals/jcss/FaginLN03,DBLP:conf/icde/AjamiC19,DBLP:conf/cikm/SadeC20,Kimelfeld:Efficiently:2008,DBLP:conf/pods/RavidMK19,DBLP:conf/sigmod/YangRLG18}. 
The computational cost of these tasks becomes very high when the number of output solutions becomes large.
Especially, enumeration has exponentially many solutions for the input size.
This makes the computational cost of enumeration very high. 
No matter what algorithm is used, the output of the solution becomes the bottleneck.
Algorithms for enumerating the $k$-best solutions are studied to overcome this obstacle.
In many applications, such as the database field, it is sufficient to find the $k$-best solutions~\cite{DBLP:conf/pods/RavidMK19,DBLP:conf/icde/AjamiC19,DBLP:journals/corr/abs-2012-09153}.
With this motivation, we believe that studying theoretically efficient $k$-best enumeration algorithms is essential.

One measure of the theoretical efficiency of an enumeration algorithm is the \emph{delay}.
In this paper, we evaluate the efficiency of $k$-best enumeration algorithms using the delay.
It is defined by the maximum time between two consecutive solutions.
We call a $k$-best enumeration algorithm a polynomial delay algorithm if a polynomial bounds its delay in the size of an input.

The task of $k$-best enumeration is to output $k$ solutions $\set{S_1, \ldots, S_k}$ such that for any other solution $S'$, $\size{S_i} \le \size{S'}$ (or $\size{S_i} \ge \size{S'}$) holds.
There are several algorithms for enumerating $k$-best solutions of (maximal) matchings, $s$-$t$ paths, and spanning trees~\cite{DBLP:journals/corr/abs-2105-04146,Murty:Letter:1968,Lawler1972,Gabow:Two:1977}.
Eppstein gave a comprehensive survey of $k$-best enumeration~\cite{DBLP:reference/algo/Eppstein16}.
Unfortunately, $k$-best enumeration problems are \NP-hard in many cases since $1$-best enumeration is equal to an optimization problem.
Therefore, it is difficult to design efficient $k$-best enumeration algorithms.
To overcome this difficulty, 
we adopt a viewpoint of approximation for $k$-best enumeration algorithms. 
We call an enumeration algorithm $\mathcal A$ a \emph{$c$-approximation $k$-best enumeration algorithm} 
if $\mathcal A$ outputs $k$-solutions $\set{S_1, \ldots, S_k}$ that satisfies $\size{S} \le c \cdot \OPT$
for any $S \in \set{S_1, \ldots, S_k}$, where $\OPT$ is the cardinality of a minimum solution which is \emph{not} outputted by the algorithm.
This definition is proposed by Fagin et al.~\cite{DBLP:journals/jcss/FaginLN03}.
They call this concept $\theta$-approximation to the top $k$ answers. 
In the field of enumeration algorithms, an algorithm for finding top $k$ answer is also called a $k$-best enumeration algorithm.
For this reason, we call this concept a $c$-approximation $k$-best enumeration algorithm.

In this paper, we consider $k$-best enumeration of the minimal connected edge dominating sets.
The enumeration of minimal (edge) dominating sets is a central topic in the field of enumeration algorithms.
Thus, without cardinality constraints, the problem has hardness results and positive results for various graph classes~\cite{Kante:Limouzy:WG:2015,DBLP:conf/fct/KanteLMN11,DBLP:journals/siamdm/KanteLMN14,DBLP:journals/talg/BonamyDHPR20}.
For minimal edge dominating set enumeration, Kant{\'e} et al.\ developed a polynomial-delay and polynomial-space algorithm~\cite{DBLP:conf/wads/KanteLMNU15}.
As a first natural question, we consider whether a polynomial-delay $k$-best enumeration algorithm exists for minimal connected edge dominating sets.
It is known that finding a minimum connected edge dominating set is \NP-hard~\cite{Munaro}(p. 102, Lemma 4.4.3).
Therefore, a polynomial-delay $k$-best enumeration of the minimal connected edge dominating set is intractable.

Another motivation for addressing $k$-best enumeration of the minimal connected edge dominating sets is the result of Kobayashi et al.~\cite{Kobayashi:Efficient:2020}.
They showed that enumeration of small minimal Steiner trees and edge dominating sets can be solved in polynomial delay with constant approximation factor.
To enumerate $k$-best solutions in an approximate manner, it must at least be possible to enumerate all solutions efficiently.
Connectivity and edge domination not only allow for efficient enumeration, but also allow for enumeration of only solutions of small cardinality.
Therefore, we studied the problem that satisfies both of two constraint simultaneously. 


\noindent\textbf{Main result:}
We show that we can enumerate $k$-best minimal connected edge dominating sets approximately with polynomial delay.
Note that our algorithm also achieves a polynomial-delay enumeration of ``all'' minimal connected edge dominating sets.
We summarize our main results in the following theorem and corollary. 

\begin{theorem}
    There is an algorithm that approximately enumerates $k$-best minimal connected edge dominating sets with a constant approximation ratio in polynomial delay.
\end{theorem}

\begin{corollary}
    There is an algorithm that enumerates all minimal connected edge dominating sets with polynomial delay.
\end{corollary}

\noindent \textbf{Related works:} 
Approximative approaches to enumeration have been studied.
In the database area, a variant of $k$-best enumeration has been introduced by Fagin et al.~\cite{DBLP:journals/jcss/FaginLN03}, called \emph{$\theta$-approximation to the top $k$ answers}. 
Ajami et al.\  proposed a $\theta$-approximation order enumeration as a more strict constraint.
For a sequence of outputs $(S_1, S_2, \ldots, S_M)$, we say that it is a $\theta$-approximation order if $\size{S_i} \le \theta \size{S_j}$ for any $i < j$.
In this problem setting, Ajami et al.\ developed a polynomial-delay $\theta$-approximation order enumeration algorithm for weighted set covers~\cite{DBLP:conf/icde/AjamiC19}.

Another approach of approximative enumeration is introduced by Kobayashi et al.\ and Agrawal et al., independently.
This approach is called \emph{approximation enumeration}~\cite{Kobayashi:Efficient:2020,DBLP:conf/icalp/AgarwalHS021}.
In this problem setting, we enumerate all solutions with weight at most a given threshold $t$.
However, we allow outputting a solution with weight at most $c t$.
We call such an enumeration algorithm a $c$-approximation enumeration algorithm.
If we have a $\theta$-approximation ordered enumeration algorithm, then it is a $\theta$-approximation enumeration algorithm of the problem.
Thus, approximation enumeration is a relaxed version of approximation order enumeration.
For a designing $c$-approximation enumeration algorithm, the \emph{CKS property} is important. See \cite{DBLP:journals/corr/abs-2004-09885,Cohen::2008} for details.
As a remarkable result, Kobayashi et al.\ show that if the monotone property has the CKS property, then there is a polynomial delay $(c+1)$-approximation enumeration algorithm~\cite{Kobayashi:Efficient:2020}.
Unfortunately, connected edge domination does not have the CKS property.

\section{Preliminaries}\label{sec:prelim}
A graph $G = (V, E)$ is a pair of the set of vertices and the set of edges.
In this paper, we assume that $G$ is connected and simple, that is, $G$ has no-self loops and multiple edges.
Thus, $\size{E} \ge \size{V} - 1$ holds.
We denote $\size{V}$ and $\size{E}$ as $n$ and $m$, respectively.
For a vertex $v \in V$, a vertex $u$ is \emph{adjacent to $v$} if $E$ has an edge $\set{u, v}$.
A vertex $u$ is called a \emph{neighbor of $v$}.
For an edge $e = \set{u, v}$, $u$ and $v$ are \emph{endpoints of $e$} and $e$ is an \emph{incident edge of $u$ and $v$}.
We denote the neighbors of $v$ and the set of incident edges of $v$ as $N_G(v)$ and $\Gamma_G(v)$, respectively.
The \emph{degree of $v$} is defined as the cardinality of $N_G(v)$, and it is denoted as $d_G(v)$.
We say that $v$ is a \emph{pendant vertex} if $\size{N_G(v)}$ is equal to one.
Similarly, we say that the edge incident to a pendent is the \emph{pendant edge}.
Moreover, the degree of $G$ is defined as $\max_{v \in V}{d_G(v)}$, and it is denoted as $\Delta_G$.
If there is no confusion arises, we drop the subscript.
The set of vertices $N(u) \cup \set{u}$ is called the \emph{closed neighbor of $u$}, and it is denoted as $N[u]$.
Moreover, for a set of vertices $U$, we define the closed neighbor of $U$ as $\bigcup_{u \in U}N(u) \setminus U$.
Similarly, the set of vertices $N(U) \cup U$ is called a \emph{closed neighbor of $U$} and it is denoted as $N[U]$.

\revise{
    A sequence of vertices $P = (u_1, u_2, \ldots, u_k)$ is called a \emph{path} if $u_i$ is adjacent to $u_{i+1}$ for any $1 \le i < k$ and all the vertices are distinct.
    A sequence of vertices $C = (u_1, u_2, \ldots, u_k)$ is called a \emph{cycle} if $u_i$ is adjacent to $u_{i+1}$ for any $1 \le i \le k$, where $u_{k + 1}$ is considered as $u_1$, and all the vertices except for pair $\{u_1, u_{k+1}\}$ are distinct.
    We say that a graph $G$ is connected if, for any pair of vertices $u, v \in V(G)$, $G$ has a $u$-$v$ path.
    A graph $G$ is called a \emph{tree} if $G$ has no cycles.
}

\revise{
    A graph $H = (U, F)$ is called a \emph{subgraph of $G$} if $U \subseteq V$ and $F \subseteq E$.
    Moreover, if $U = \bigcup_{f \in F}f$, then $H$ is called an \emph{edge induced subgraph}. 
    We denote an edge induced subgraph $H = (U, F)$ as $G[F]$.
    If a subgraph $H$ is a tree, that is, $H$ has no cycles, then $H$ is called a \emph{subtree of $G$}.
    We call a graph $H = (U, F)$ of $G = (V, E)$ is an \emph{induced subgraph of $G$} if $U \subseteq V$ and $F = \inset{\set{u, v} \in E}{u, v \in U}$.
}

For edges $e$ and $f$, we say that \emph{$e$ dominates $f$} if $e$ and $f$ share one endpoint.
Similarly, for vertices $v$ and $u$, we say that \emph{$v$ dominates $u$} if $u$ is a neighbor of $v$.
A set of edges $F$ is called \emph{edge dominating set} if  any edge in $G$ is dominated by an edge in $F$. 
Similarly, for a set of vertices $U$, $U$ is called \emph{dominating set} if $N[U]$ is equal to $V$.
Moreover, $F$ is called \emph{connected edge dominating set} if $G[F]$ is connected and $F$ is an edge dominating set of $G$.
For a connected edge dominating set $F$, an edge $e = \set{u, v} \not\in F$ is a \emph{private edge} of $f \in F$ if $(\Gamma(u) \cup \Gamma(v)) \cap F = \set{f}$.
A connected edge dominating set $F$ is called a \emph{minimal connected edge dominating set} if for any $f \in F$, $F \setminus \set{f}$ is not a connected edge dominating set.
We can easily see that the following proposition holds. 

\begin{proposition}
Let $F$ be a connected edge dominating set. 
Then, $G[F]$ forms a tree if $F$ is minimal. 
\end{proposition}

\section{Enumeration of all minimal connected edge dominating sets}
To achieve polynomial-delay enumeration of $k$-best minimal connected edge dominating sets, 
we first give a polynomial-delay enumeration algorithm for ``all'' minimal connected edge dominating sets. 
In this paper, we assume that the cardinality of a minimum connected edge dominating set is at least two. 
If $G$ has a minimum connected edge dominating set with cardinality one, then the enumeration of minimal connected edge dominating sets can be done in $\order{\Delta m}$ time.

We first consider a trivial case. The cardinality of a minimum connected edge domination is one.
We can enumerate all solutions in $\order{m\Delta}$ time in this case.
Let $e^* = \set{x, y}$ be a minimum connected edge dominating set.
Since $e^*$ is an edge dominating set, any edge in $G$ is incident to at least one endpoint of $e^*$,
and every minimal connected edge dominating set contains at most one edge incident to each endpoint of $e^*$.
Such a graph has only three types of solutions as follows,
(I)   a solution $\set{e^*}$,
(II)  a solution that contains edges $f$ incident to $x$ and $g$ incident to $y$.
(III) a solution such that all edges in the solution are incident to either vertex $x$ or $y$.
The type-I is unique, the number of type-II solutions is at most $n$, and the number of type-III solutions is at most $2$.
Thus, we can enumerate all minimal connected edge dominating set in $\order{m\Delta}$ time.
See \Cref{fig:trivial} for an example.

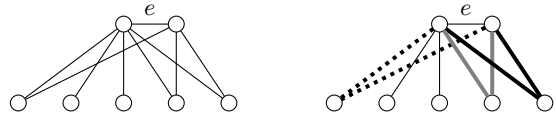
\begin{figure}[t]
    \centering
    \begin{tikzpicture}

\tikzset{neigh node/.style={
circle, fill=white, draw=black, 
minimum size=#1, inner sep=0pt, outer sep=0pt}, 
neigh node/.default = 6pt}

\begin{scope}[scale=.7]
		\node [neigh node] (u)  at (1, 1.5) {};
		\node [neigh node] (v)  at (2, 1.5) {};

		\node [neigh node] (a)  at (-1, 0) {};
		\node [neigh node] (b)  at (0, 0) {};
		\node [neigh node] (c)  at (1, 0) {};
		\node [neigh node] (d)  at (2, 0) {};
		\node [neigh node] (e)  at (3, 0) {};

		\node [neigh node] (u')  at (7, 1.5) {};
		\node [neigh node] (v')  at (8, 1.5) {};

		\node [neigh node] (a')  at (5, 0) {};
		\node [neigh node] (b')  at (6, 0) {};
		\node [neigh node] (c')  at (7, 0) {};
		\node [neigh node] (d')  at (8, 0) {};
		\node [neigh node] (e')  at (9, 0) {};

		\node [neigh node] (u'')  at (13, 1.5) {};
		\node [neigh node] (v'')  at (14, 1.5) {};

		\node [neigh node] (a'')  at (11, 0) {};
		\node [neigh node] (b'')  at (12, 0) {};
		\node [neigh node] (c'')  at (13, 0) {};
		\node [neigh node] (d'')  at (14, 0) {};
		\node [neigh node] (e'')  at (15, 0) {};

    	\draw (u) -- node[above] {$e^*$} (v);
    	\draw (v) -- (a);
        \draw (v) -- (d);
        \draw (v) -- (e);
        \draw (u) -- (a);
        \draw (u) -- (b);
        \draw (u) -- (c);
        \draw (u) -- (e);
        \draw (u) -- (d);
        
    	\draw (u') -- node[above] {$e^*$} (v');
    	\draw[dotted, line width = 1.5pt] (v') -- (a');
        \draw[draw = gray, line width = 1.5pt] (v') -- (d');
        \draw[line width = 1.5pt] (v') -- (e');
        \draw[dotted, line width = 1.5pt] (u') -- (a');
        \draw (u') -- (b');
        \draw (u') -- (c');
        \draw[line width = 1.5pt] (u') -- (e');
        \draw[draw = gray, line width = 1.5pt] (u') -- (d');

    	\draw (u'') -- node[above] {$e^*$} (v'');
            \draw (v'') -- (a'');
        \draw (v'') -- (d'');
        \draw (v'') -- (e'');
        \draw[draw = gray, line width = 1.5pt] (u'') -- (a'');
        \draw (u'') -- (b'');
        \draw (u'') -- (c'');
        \draw[draw = gray, line width = 1.5pt] (u'') -- (e'');
        \draw[draw = gray, line width = 1.5pt] (u'') -- (d'');
\end{scope}
\end{tikzpicture}
    \caption{An example of a trivial case. The left graph has a minimum connected edge dominating set with cardinality one. 
    The type-II solutions are three minimal solutions with two edges represented by dotted edges, gray edges, and black edges in the middle graph. 
    \revise{The type-III solution is one solution with three edges represented by gray edges in the right graph.}}
    \label{fig:trivial}
\end{figure}

To design a polynomial-delay enumeration algorithm, 
we adopt \emph{supergraph technique}~\cite{Kobayashi:Efficient:2020,DBLP:journals/corr/abs-2105-04146} as the basic idea. 
In this technique, we define a directed graph $\mathcal G = (\mathcal V, \mathcal E)$, called a \emph{supergraph}, whose vertex set $\mathcal V$ consists of the set of all minimal connected edge dominating sets in an input graph.
If $\mathcal G$ is strongly connected, then by starting from an arbitrary solution, 
we can enumerate solutions by applying a standard graph traverse procedure on $\mathcal G$. 
Thus, the key of supergraph technique is how we define the edge set $\mathcal E$ of $\mathcal G$.

We define directed edges in $\mathcal G$ by three types of neighbors for each minimal connected edge dominating set $X$.
The common idea is simple.
We firstly remove an edge $e$ from $X$, next add several edges to $X \setminus \set{e}$, and then,
obtain a connected edge dominating set $X'$. 
To obtain a minimal connected edge dominating set that is contained in $X'$, we use $\comp{\cdot}$ defined as follows: 
$\comp{\cdot}$ is an arbitrary deterministic procedure 
such that $\comp{\cdot}$ receives a connected edge dominating set $X$ and outputs a minimal connected edge dominating set $Y$ contained in $X$. 
The following proposition shows $\comp{\cdot}$ can be computed in linear time. 
The following characterization of minimal edge connected dominating sets is important for $\comp{\cdot}$. 

\begin{proposition}\label{prop:pendant:has:private}
    Let $X$ be a connected edge dominating set of a graph $G$.
    Then, $X$ is minimal if and only if $G[X]$ is a tree and any pendant edge has at least one private edge.
\end{proposition}
\begin{proof}
    Suppose that $X$ is minimal.
    Since $X$ is a minimal solution, $G[X]$ is a tree. 
    Otherwise, we can remove an edge that is not a bridge from $X$.
    We assume that $X$ has a pendant edge $e$ having no private edges. 
    Then, $X\setminus \set{e}$ dominates all edges.
    Moreover, since $e$ is a pendant edge of $X$, $G[X\setminus\set{e}]$ is connected.
    It contradicts the assumption and completes the if part. 

    We next show the only if part. 
    We prove this by contradiction.
    Let $Y$ be a connected edge dominating set \revise{strictly contained in $X$, that is, $Y \subset X$}.
    If $Y$ contains all pendant edges in $X$, then $Y$ is not a tree since any non-pendant edge is a bridge. 
    Thus, $Y$ does not contain some pendant edge in $X$. 
    However, it contradicts that $Y$ is an edge dominating set. 
\end{proof}

The following lemma shows that $\comp{\cdot}$ can be done in linear time.

\begin{lemma}\label{lem:comp}
    Let $X$ be a connected edge dominating set.
    We can compute a minimal connected edge dominating set contained in $X$ in $\order{m}$ time.
\end{lemma}
\begin{proof}
    We first compute a spanning tree $T$ of $G[X]$. 
    Clearly, $E(T)$ is a connected edge dominating set in $G$ and can be obtained in linear time. 
    We show that, for any subset $T' \subseteq T$, if $T'$ is a connected edge dominating set, then
    $T'$ contains all non-pendant edges in $T$. 
    For any \revise{non-pendant} edge $e \in T$, $G[T \setminus \set{e}]$ has two connected components $C_1$ and $C_2$.
    Both $E(C_1)$ and $E(C_2)$ are not edge dominating sets for $G$ since any edge in $C_1$ is not dominated by $C_2$ and vice versa.
    Thus, we cannot remove any non-pendant edge.

    In the following, we determine whether we can remove from $T$ an edge $e$ with a leaf of $T$ as an endpoint.
    Let $\ell$ be a pendant vertex in $G[X]$ that has $e$ as an incident edge.
    If the other endpoint of any edge in $\Gamma_G(\ell)$ is in $V(T)$, 
    that is, $e$ has no private edges, 
    then
    we remove $e$ from $T$ and set $T = T \setminus \set{e}$. 
    Otherwise, we keep $e$ in $T$. 
    This can be done in $\order{d(\ell)}$ time. 
    We repeat this procedure until all pendant edges in $T$ are examined. 
    Note that any pendant edge is examined at most once. 
    Any pendant edge in the resultant tree $T$ has a private edge. 
    Thus, by \Cref{prop:pendant:has:private}, $T$ is minimal and  the total time complexity is $\order{m}$. 
\end{proof}

Next, we define three types of neighbors of a minimal connected edge dominating set $X$. 
For any pendant edge $e$ in $G[X]$, we denote by $\wpp{X, e}$ the set of vertices such that these are endpoints of a private edge $h$ of $e$.

\begin{description}
    \item[Type-I neighbor:]
    Suppose that $e$ is a bridge in $G[X]$ and $f = \set{u, v}$ is an edge, where $u \in V(G[X\setminus\set{e}])$ and $v \notin V(G[X\setminus\set{e}])$. 
    Let $C_1$ and $C_2$ be the two connected components in $G[X \setminus\set{e}]$.
    Without loss of generality, $f$ incidents to $V(C_1)$ and $v$ is not contained in $V(G[X \setminus\set{e}])$.
    Let $P$ be the subset of $\Gamma_G(v)$ such that for each edge $g$ in $P$, the other endpoint of $g$ is in $V(C_2)$. 
    Then, for each $g \in P$, 
    we say that $\comp{(X\setminus \set{e}) \cup \set{f, g}}$ is a \emph{type-I neighbor of $X$ with respect to the edge pair $(e, f)$}.
    Note that $f$ may equal to $g$. 
    
    \item[Type-II neighbor:]
    Suppose that $e = \set{u ,v}$ is a pendant edge in $G[X]$, where 
    $v$ is a pendant vertex in $G[X]$.
    Let $\mathcal P$ be the set of all paths from $v$ to $V(G[X\setminus\set{e}])$ in $G$ with length at most two.
    For each $P \in \mathcal P$,
    we say $\comp{(X\setminus \set{e}) \cup E(P)}$ is a \emph{neighbor-II of $X$ with respect to $e$}.

    \item[Type-III neighbor:]
    Suppose that $e = \set{u ,v}$ is a pendant edge in $G[X]$ such that $G$ has no pendant edges incident to $v$, where 
    $v$ is a pendant vertex in $G[X]$.
    Let $F$ be an edge set obtained as follows: for each $w$ in $\wpp{X, e}$, 
    we pick an edge $f$ between $w$ and $V(G[X \setminus\set{e}])$ and add $f$ to $F$. 
    Then, we say that $\comp{(X\setminus\set{e}) \cup F}$ is the \emph{neighbor-III of $X$ with respect to $e$}.
\end{description}

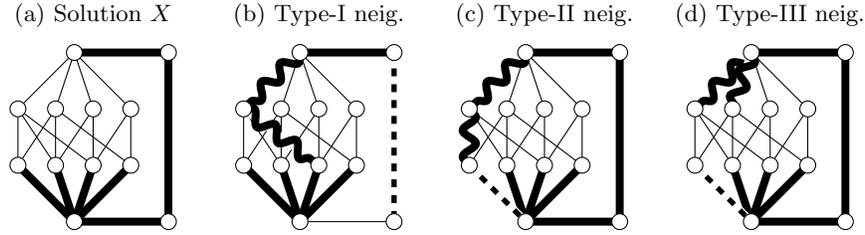
\begin{figure}[t]
    \centering
    \begin{tikzpicture}
\tikzset{neigh node/.style={
circle, fill=white, draw=black, 
minimum size=#1, inner sep=0pt, outer sep=0pt}, 
neigh node/.default = 6pt}
\tikzset{base node/.style={inner sep=0pt, outer sep=0pt, minimum size=0pt}}
\tikzset{tree cover edge/.style={line width=3pt}}
\tikzset{added/.style={decorate, decoration=snake, line width=3pt}}
\tikzset{removed/.style={dashed, line width=2pt}}

\newcommand{\neighscale}{.5}
\newcommand{\neighfigtitle}[1]{\node[base node] (label) at (2, 5.5) {#1};}
\newcommand{\neighbasegraph}{
    \node[base node]  (u1) at (0, 1.5)   {};
    \node[base node]  (u2) at (1, 1.5)   {};
    \node[base node]  (u3) at (2, 1.5)   {};
    \node[base node]  (u4) at (3, 1.5)   {};
    \node[base node]  (v1) at (0, 3)     {};
    \node[base node]  (v2) at (1, 3)     {};
    \node[base node]  (v3) at (2, 3)     {};
    \node[base node]  (v4) at (3, 3)     {};
    \node[base node]  (w)  at (1.5, 4.5) {};
    \node[base node]  (x)  at (4, 4.5)   {};
    \node[base node]  (y)  at (4, 0)     {};
    \node[base node]  (z)  at (1.5, 0)   {};
    \foreach \u / \v in {z/u1, z/u2, z/u4, z/u3,
        v1/u1, v1/u2, v1/u3, v2/u1, v2/u2, v2/u4, v3/u2, v3/u3, v4/u3, v4/u4,
        w/v1, w/v2, w/v3, w/v4,
        x/w, y/x, z/y} 
    {
        \draw (\u) -- (\v);
    }
}

\begin{scope}[scale=\neighscale]
    \neighfigtitle{(a) Solution $X$}
    
    \neighbasegraph

    \foreach \u / \v in {w/x, x/y, y/z, z/u1, z/u2, z/u3, z/u4}
    {
        \draw[tree cover edge] (\u) -- (\v);
    }
    
    \foreach \u in {u1,u2,u3,u4,v1,v2,v3,v4,w,x,y,z} {
        \node[neigh node] at (\u) {}; 
    }
\end{scope}

\begin{scope}[scale=\neighscale, shift={(6,0)}]
    \node[base node] (label) at (2, 5.5) {(b) Type-I neig.}; 
    \neighbasegraph

    \foreach \u / \v in {w/x, z/u1, z/u2, z/u4, z/u3}
    {
        \draw[tree cover edge] (\u) -- (\v);
    }
    
    \foreach \u / \v in {w/v1, v1/u3}
    {
        \draw[draw=white, line width = 3pt] (\u) -- (\v);
        \draw[added] (\u) -- (\v);
    }
    \foreach \u / \v in {x/y}
    {
        \draw[draw=white, line width = 3pt] (\u) -- (\v);
        \draw[removed] (\u) -- (\v);
    }
    
    \foreach \u in {u1,u2,u3,u4,v1,v2,v3,v4,w,x,y,z} {
    \node[neigh node] at (\u) {}; 
    }
\end{scope}

\begin{scope}[scale=\neighscale, shift={(12,0)}]
    \node[base node] (label) at (2, 5.5) {(c) Type-II neig.}; 
    \neighbasegraph

    \foreach \u / \v in {w/x, x/y, y/z, z/u2, z/u3, z/u4}
    {
        \draw[tree cover edge] (\u) -- (\v);
    }
    
    \foreach \u / \v in {w/v1, v1/u1}
    {
        \draw[draw=white, line width = 3pt] (\u) -- (\v);
        \draw[added] (\u) -- (\v);
    }
    \foreach \u / \v in {u1/z}
    {
        \draw[draw=white, line width = 3pt] (\u) -- (\v);
        \draw[removed] (\u) -- (\v);
    }
    
    \foreach \u in {u1,u2,u3,u4,v1,v2,v3,v4,w,x,y,z} {
    \node[neigh node] at (\u) {}; 
    }
\end{scope}

\begin{scope}[scale=\neighscale, shift={(18,0)}]
    \node[base node] (label) at (2, 5.5) {(d) Type-III neig.}; 
    \neighbasegraph

    \foreach \u / \v in {w/x, x/y, y/z, z/u2, z/u3, z/u4}
    {
        \draw[tree cover edge] (\u) -- (\v);
    }
    \foreach \u / \v in {w/v1, w/v2}
    {
        \draw[draw=white, line width = 3pt] (\u) -- (\v);
        \draw[added] (\u) -- (\v);
    }
    
    \foreach \u / \v in {u1/z}
    {
        \draw[draw=white, line width = 3pt] (\u) -- (\v);
        \draw[removed] (\u) -- (\v);
    }
    
    \foreach \u in {u1,u2,u3,u4,v1,v2,v3,v4,w,x,y,z} {
    \node[neigh node] at (\u) {}; 
    }
\end{scope}
\end{tikzpicture}
    \caption{Examples for three types of neighbors of a minimal connected edge dominating set $X$. 
    (a) depicts $X$. Black bold lines indicate the edges in $X$. 
    (b) depicts a type-I neighbor $Y_1$ of $X$. 
    $Y_1$ is obtained by removing a dashed edge corresponding $e$ and adding a wavy path corresponding $\set{f, g}$. 
    The bottom edge is removed by applying $\comp{\cdot}$. 
    (c) and (d) are a type-II and a type-III neighbor of $X$, respectively. 
    Similarly, the neighbors are obtained by removing a dashed edge and adding wavy edges. 
   }
    \label{fig:trans}
\end{figure}

See \Cref{fig:trans} for examples. 
Note that every vertex in $\wpp{X, e}$ in the case when we are considering the Type-III neighbor has at least two neighbors.
Intuitively speaking, to obtain a type-I neighbor,  we first destroy the connectivity constraint, and then add a path $(f, g)$ so that the resultant edge set is connected again. 
Similarly, to obtain a type-II or type-III neighbor, we first destroy the domination constraint, and then add some edges so that the resultant edge set satisfies the domination constraint. 
By the following lemmas, the obtained edge sets are actually minimal connected edge dominating sets. 

The following auxiliary lemma is useful to show \Cref{lem:nei:edom}.
    \begin{lemma}\label{lem:pendant}
        Let $X$ be a minimal connected edge dominating set and  
        $e = \set{u ,v}$ be a pendant edge in $G[X]$ such that $v$ is a pendant vertex in $G[X]$.    
        Then, for any vertex $w$ in $\wpp{X, e}$, $\Gamma_G(w)$ has an edge that connects $w$ and $V(G[X \setminus \set{e}])$.
    \end{lemma}
    \begin{proof}
        Let $f$ be an edge $\set{v, w}$.
        Since $f$ is not a pendant edge, $\Gamma(w)$ has an edge $g \neq f$.
        Since $X$ is a connected edge dominating set and $e$ is a pendant edge, 
        $g$ is dominated by an edge in $X$ both of whose endpoints are not $v$. 
        Thus, $g$ connects $w$ and $V(G[X\setminus\set{e}])$. 
    \end{proof}

\begin{lemma}\label{lem:nei:edom}
    Let $X$ be a minimal connected edge dominating set and $Y$ be a neighbor of $X$. 
    Then, $Y$ is a minimal connected edge dominating set. 
\end{lemma}
\begin{proof}
    Suppose that $Y$ is a type-I neighbor of $X$ with respect to $(e, f)$. 
    Since $e$ is a bridge, $V(G[X])$ is equal to $V(G[X\setminus \set{e}])$. 
    Moreover, from the definition of $g$, $G[(X \setminus \set{e}) \cup \set{f, g}]$ is connected. 
    Thus, $(X \setminus \set{e}) \cup \set{f} \cup E(P)$ is a connected edge dominating set.  
    
    Suppose that $Y$ is a type-II neighbor of $X$ with respect to $e = \set{u, v}$, where $v$ is a pendant vertex in $G[X]$. 
    Since we add $E(P)$ that is incident to $v$, $V(G[(X\setminus \set{e}) \cup E(P)]) \supseteq V(G[X])$.
    Note that $E(P)$ may contain $e$. 
    Thus, $(X \setminus \set{e}) \cup E(P)$  is a connected edge dominating set. 
    
    Suppose that $Y$ is a type-III neighbor of $X$ with respect to $e$. 
    Let $h$ be a private edge of $e$. 
    From the definition of $Y$, there is an edge $g$ in $F$ such that $h$ and $g$ share a vertex in $\wpp{X, e}$. 
    Moreover, by \Cref{lem:pendant}, $g$ is dominated by $X \setminus \set{e}$. 
    Thus, $(X \setminus \set{e}) \cup F$  is a connected edge dominating set. 
    In each case, $Y$ is obtained by applying $\comp{\cdot}$ to a connected edge dominating set. 
    Hence, the lemma holds. 
\end{proof}

Let $\neigh{X}$ be the set of all types of neighbors of $X$. 
We say that $Y$ is a \emph{neighbor} of $X$ if $Y \in \neigh{X}$. 
Now, we define directed edges in $\mathcal G$ as follows: There is a directed edge $(X, Y)$ in $\mathcal G$ if and only if $Y \in \neigh{X}$. 
Note that $X \in \neigh{Y}$ may not hold. 
We next show the strong connectivity of $\mathcal G$.
To this end, for a pair of two minimal connected edge dominating sets $X$ and $Y$, we define a function $\distor{X}{Y}$ as the size of the union $\distor{X}{Y} \coloneqq \size{X \cup Y}$ of $X$ and $Y$. 
Intuitively speaking, this function expresses the ``closeness'' of $X$ and $Y$. 
The following proposition is easily derived but essential. 

\begin{proposition}\label{prop:close}
    For any pair of minimal connected edge dominating sets $X$ and $Y$, $\distor{X}{Y} = \size{X} = \size{Y}$ if and only if $X = Y$.    
\end{proposition}

Thus, to show the strong connectivity of $\mathcal G$, it is sufficient to show that
for any pair of minimal connected edge dominating sets $X$ and $Y$, 
there exists a minimal connected edge dominating set $Z$ such that there is a directed path from $X$ to $Z$ in $\mathcal G$ satisfying $\distor{X}{Y} > \distor{Z}{Y}$.
That is, $Z$ is ``closer'' to $Y$ than $X$ with respect to $\distor{\cdot}{\cdot}$. 
We consider the following two cases to show that such a $Z$ exists.
One is the case such that $G[X \cup Y]$ has at least one cycle and the other is the case such that $G[X \cup Y]$ has no cycles.
We consider the case such that $G[X \cup Y]$ has at least one cycle $C$.
We show that  $\mathcal G$ contains a directed edge $(X, Z)$ in this case.

Let $C$ be a cycle in $G[X \cup Y]$. Since both $X$ and $Y$ are trees, $C$ contains both edges in $X$ and $Y$. 
We show that $C$ can be decomposed into two parts as follows.

\begin{lemma}\label{lem:cycle}
    Let $X$ and $Y$ be minimal connected edge dominating sets. 
    Suppose that $G[X \cup Y]$ contains a cycle. 
    Then, $G[X \cup Y]$ has a cycle $C = \set{v_1, \dots, v_k}$ such that 
    (1) $C$ can be decomposed into two paths $P$ and $Q$ that satisfies $E(P) \subseteq X$, $E(Q) \subseteq Y$, and $E(Q) \cap X = \emptyset$.
    (2) the length of $Q$ is at most two, and  
    (3) any non-endpoint vertex of $Q$ is not in $V(G[X])$. 
\end{lemma}
\begin{proof}
    Without loss of generality, 
    $v_1$ is contained in $V(G[X])$ and 
    $C$ can be decomposed into the sequence of paths $(R_1, \ldots, R_\ell)$ such that 
    (1) all edges in $R_i$ are contained in $X$ if $i$ is odd, and  
    (2) all edges in $R_i$ are contained only in $Y$ if $i$ is even. 
    We consider $R_2$.
    \revise{
        If the length of $R_2$ is one, that is $R_2 = \set{r_1, r_2}$,
        we obtain a desired cycle by combining $r_1$-$r_2$ path in $G[X]$.
        
        Suppose that the length of $R_2 = (r_1, \ldots, r_\ell)$ is at least two.
        Since $E(R_2) \subseteq Y$, $E(R_2) \cap X = \emptyset$, and $X$ is an edge dominating set, 
        $X$ has an edge that is incident to $r_3$.
        Otherwise, an edge $\set{r_2, r_3}$ is not dominated by $X$, which contradicts that $X$ is an edge dominating set. 
        Therefore, $G[X]$ contains $r_1$ and $r_3$.
        Moreover, since $G[X]$ is connected, $G[X]$ has a $r_1$-$r_3$ path $P$.
        We obtain a cycle $C'$ by combining a $P$ and a path $(r_1, r_2, r_3)$.
        This cycle satisfies the condition (1) and (2).
        Finally, since $E(R_2) \cap X = \emptyset$, $G[X]$ does not contain $r_2$ and $C'$ is a desired cycle.
    }
\end{proof}

From the above lemma, $C$ can be decomposed into two paths $P$ and $Q$.
Since the length of $Q$ is at most two, 
we can show that $X$ has a type-I neighbor or a type-II neighbor that is closer to $Y$ when an edge in $E(P)$ is selected.

\begin{lemma}\label{lem:edge:in:cycle:is:removable}
    Let $X$ and $Y$ be a pair of minimal connected edge dominating sets such that $X \neq Y$.
    Suppose that $G[X \cup Y]$ has a cycle $C$. 
    For any edge $e$ in $(E(C) \cap X) \setminus Y$, 
    $X$ has a neighbor $Z$ contained in $(X \cup Y) \setminus \set{e}$.
\end{lemma}
\begin{proof}
    Let $C$ be a cycle in $G[X \cup Y]$ that satisfies the conditions in \Cref{lem:cycle}. 
    Thus, $C$ can be respectively decomposed into two paths $P$ and $Q$ contained in $X$ and $Y$, 
    and the length of $Q$ is at most two. 
    In what follows, we assume that $e$ is an edge in $P$. 
    
    Suppose that $e$ is a bridge in $G[X]$.
    Let $f$ be an edge in $Q$. 
    Note that $f$ is incident to one of the connected components in $G[X \setminus \set{e}]$. 
    From the definition of a type-I neighbor, 
    $X$ has a neighbor $Z = \comp{(X\setminus\set{e}) \cup \set{f} \cup (E(Q) \setminus \set{f})}$.
    Since $Q$ is contained in $Y$, $Z$ is a subset of $(X \cup Y) \setminus \set{e}$.
    
    Suppose that the remaining case, that is, $e = \set{u, v}$ is a pendant edge in $G[X]$. 
    Without loss of generality, $v$ is a pendant vertex in $G[X]$.
    Then, $Q$ contains an edge incident to $v$.
    Hence, from the definition of a type-II neighbor and $Q \subseteq Y$, 
    $Z = \comp{(X\setminus \set{e}) \cup E(Q)}$ is a neighbor of $X$, 
    and $Z$ is a subset of $(X \cup Y) \setminus \set{e}$. 
\end{proof}

From the above lemma, we can immediately obtain the following lemma. 

\begin{lemma}\label{lem:connected}
    Let $X$ and $Y$ be a pair of minimal connected edge dominating sets such that $X \neq Y$.
    Suppose that $G[X \cup Y]$ has a cycle. 
    Then, $X$ has a neighbor $Z$ that satisfies $\distor{X}{Y} > \distor{Z}{Y}$. 
\end{lemma}

Next, we need to show the remaining case that $G[X \cup Y]$ has no cycles.
Note that if $G[X \cup Y]$ has at least two connected components, it contradicts that both $X$ and $Y$ are connected edge dominating sets since every edge $e \in X$ is not dominated by $Y$.
Thus, $G[X \cup Y]$ is a tree.
Let $e$ be an edge $X\setminus Y$.
The following lemma shows that $e$ is a pendant edge in $G[X]$. 
\begin{lemma}\label{lem:XY:pendant}
    Let $X$ and $Y$ be minimal connected edge dominating sets.
    If $G[X \cup Y]$ is a tree, then any edge in $X \setminus Y$ is a pendant edge in $G[X]$.
\end{lemma}
\begin{proof}
    We prove the statement of the lemma using contradiction
    Suppose that $e = \set{u, v} \in X \setminus Y$ is not a pendant edge in $G[X]$. 
    Without loss of generality, $Y$ has an edge incident to $v$.
    Since $G[X\cup Y]$ is a tree, $Y$ has no edges incident to $u$. 
    However, $X$ has an edge $f = \set{u, w}$ incident to $u$ since $e$ is not a pendant edge in $G[X]$.
    Since $Y$ is an edge dominating set, $Y$ has an edge $\set{w, x}$ incident to $w$.
    Since $Y$ is a connected edge dominating set, $G[Y]$ has a $w$-$v$ path that does not contain $e$
    Thus, $G[X \cup Y]$ has at least two $w$-$v$ paths.
    it contradicts that $G[X \cup Y]$ is a tree, and the statement holds.
\end{proof}

Let $Z'$ be the type-III neighbor of $X$ with respect to $e \in X \setminus Y$ and $w$ be a pendant vertex in $G[X]$ that has $e$ as an incident edge.
Since $G[X \cup Y]$ is a tree and $Y$ is an edge dominating set, $w$ has no pendant edges in $G$.
Otherwise, it contradicts that $Y$ is an edge dominating set.
Thus, we can pick $e$ as an edge for obtaining a type-III neighbor of $X$.
However, $\distor{Z'}{Y} \ge \distor{X}{Y}$ may hold since the added edge set $F \subseteq Y$ may not holds. 
In \Cref{lem:tree}, we show that we can remove edges in $F \setminus Y$ from $Z'$ by repeatedly obtaining type-I or type-II neighbors. 

The outline of our proof of \Cref{lem:tree} is as follows. 
We first consider a minimal connected edge dominating set $T \subseteq (X \cup Y \cup F) \setminus \set{e}$. 
We show that $T$ has a type-I or type-II neighbor $T' \subseteq (X \cup Y \cup (F \cap T)) \setminus \set{e, f}$, where $f$ is an edge in $F \setminus Y$.
If the above claim holds, then there exists a minimal connected edge dominating set $Z$
that satisfies $\distor{Z}{Y} < \distor{X}{Y}$
by applying the claim recursively.
Our main idea to show the claim is to show that either $G[T \cup X]$ or $G[T \cup Y]$ has a cycle that contains $f$.
The following auxiliary lemmas are useful to show the key lemma.

\begin{lemma}\label{lem:pendant:cycle}
    Let $X$ and $Y$ be two minimal connected edge dominating sets, $v$ be a pendent vertex in $G[X]$,
    $e_X$ be a pendant edge $\set{v, a}$ in $G[X]$, and $e_Y$ be an edge $\set{v, b}$ in $G[Y]$.
    If $e_X \in X \setminus Y$ and $e_Y \in Y \setminus X$, then $G[X \cup Y]$ contains at least one cycle that contains $e_X$.
\end{lemma}
\begin{proof}
    If $Y$ has an edge incident to $a$, then $G[Y]$ has a $a$-$b$ path, and the statement holds.
    Suppose that $Y$ has no edges incident to $a$. 
    Since $\size{X}$ is at least two, $a$ is not a pendant vertex in $G[X]$.
    Thus, $\Gamma_{G[X]}(a) \setminus \set{e_X}$ has an edge $f = \set{a, c}$.
    Since $Y$ is an edge dominating set and $Y$ has no edges incident to $a$, $Y$ has an edge incident to $c$.
    Thus, $G[Y]$ has a $v$-$c$ path, and the statement holds.
\end{proof}

    
Now, we are ready to prove the following key lemma.

\begin{lemma}\label{lem:tree}
    Let $X$ and $Y$ be two distinct minimal connected edge dominating sets.
    Suppose that $G[X \cup Y]$ is a tree.
    Then, $\mathcal G$ has a directed path from $X$ to $Z$, where $Z$ is a minimal connected edge dominating set satisfying $\distor{X}{Y} > \distor{Z}{Y}$.
\end{lemma}
\begin{proof}
    Let $e$ be an edge in $X \setminus Y$.
    From \Cref{lem:XY:pendant}, $e$ is a pendant edge in $G[X]$. 
    Let $v$ be a pendant vertex of $G[X]$ that has $e$ as an incident edge.

    Let $Z'$ be the type-III neighbor of $X$ with respect to $e$, 
    and let $F$ be the set of edges that is added to $X$ for obtaining $Z'$.
    From the construction, $Z'$ is a subset of $(X \cup Y \cup F) \setminus \set{e}$.
    We show that for a minimal connected edge dominating set $T \subseteq (X \cup Y \cup F) \setminus \set{e}$,
    $T$ has a neighbor $T' \subseteq (X \cup Y \cup (F \cap T)) \setminus \set{e, f}$ 
    for some edge $f$ in $(F \setminus Y) \cap T$. 
    If this holds, then by repeatedly obtaining such neighbors, 
    we can show that $\mathcal G$ has a directed path from 
    $Z'$ to a minimal connected edge dominating set $Z$ satisfying $Z \subseteq (X \cup Y) \setminus \set{e}$  
    since both $Z'$ and $T'$ are subsets of $(X \cup Y \cup F) \setminus \set{e}$. 
    
    Let $f$ be an edge in $(F \setminus Y) \cap T$. 
    Suppose that $f$ is a non-pedant edge in $G[T]$. 
    Then, $G[T \cup Y]$ contains a cycle $C$ such that $f \in E(C)$ by \revise{the contraposition of \Cref{lem:XY:pendant}.}
    Thus, by \Cref{lem:edge:in:cycle:is:removable}, there is a neighbor $T'$ satisfying $T' \subseteq (T \cup Y) \setminus \set{f} \subseteq (X \cup Y \cup (F \cap T)) \setminus \set{e, f}$. 
    
    We next suppose that $f = \set{a, w}$ is a pendant edge in $G[T]$, where $w$ is a pendant vertex in $G[T]$.
    Since $f$ is contained in $F$, $w$ is contained in $\wpp{X, e}$. 
    We firstly show that the following claim holds. 
    \begin{claim}
        $\Gamma(a) \cap X \neq \emptyset$ and $\Gamma(w) \cap Y \neq \emptyset$.
    \end{claim}
    \begin{claimproof}
        If $\Gamma(w) \cap X \neq \emptyset$, then this contradicts that $e$ has a private edge whose endpoint is $w$.
        Thus, $\Gamma(a) \cap X \neq \emptyset$ holds since $f$ is dominated by $X$. 
        
        If $\Gamma(w) \cap Y = \emptyset$, then to dominate a private edge $h$ of $e$ by $Y$ such that $h = \set{v, w}$, $Y$ contains an edge whose endpoint is $v$. 
        Recall that $e \notin Y$, $e \in X$, and $e$ is a pendant edge in $G[X]$. 
        This implies that by \Cref{lem:pendant:cycle}, $G[X \cup Y]$ has a cycle and a contradiction occurs. 
        Thus, $\Gamma(w) \cap Y \neq \emptyset$. 
    \end{claimproof}
    
    Recall that $f \notin Y$. 
    From \Cref{lem:pendant:cycle} and the above claim, $G[T \cup X]$ or $G[T \cup Y]$ contains a cycle $C$ such that $f \in E(C)$.
    Hence, from \Cref{lem:edge:in:cycle:is:removable}, $T$ has a neighbor $T'$ that is contained in either $(T \cup X) \setminus \set{f}$ or 
    $(T \cup Y) \setminus \set{f}$.
    Since $T$ is contained in $(X \cup Y \cup F) \setminus \set{e}$, 
    $T'$ is contained in $(X \cup Y \cup (F \cap T)) \setminus \set{e, f}$. 
    Therefore, the lemma holds. 
\end{proof}

From \Cref{lem:tree,lem:connected}, the following lemma holds immediately.

\begin{lemma}\label{lem:edom:scc}
    For any pair of minimal connected edge dominating sets $X$ and $Y$ such that $X \neq Y$, 
    $\mathcal G$ has a directed path from $X$ to $Z$, where $Z$ is a minimal connected edge dominating set satisfying $\distor{X}{Y} > \distor{Z}{Y}$. 
\end{lemma}

\Cref{prop:close} and \Cref{lem:edom:scc} imply that $\mathcal G$ is strongly connected. 
Thus, by using a standard graph search algorithm, we can traverse all vertices in $\mathcal G$ starting from an arbitrary solution.
\Cref{algo:enum:all} depicts the pseudocode. 

\DontPrintSemicolon
\begin{algorithm}[t]
    \caption{An algorithm enumerates all minimal connected edge dominating sets in $\order{nm^2\Delta}$ delay.}
    \label{algo:enum:all}
    \Procedure{\MTC{$G$}}{
        $T^* \gets$ an arbitrary minimal connected edge dominating set $G$\;
        Push $T^*$ to an empty queue $\mathcal Q$ and add $T^*$ to an empty set $\mathcal S$\;
        \While{$\mathcal Q$ is not empty}{
            $T \gets$ a minimal connected edge dominating set in $\mathcal Q$\;
            Output $T$ and delete $T$ from $\mathcal Q$\;
            \ForEach{neighbor $T'$ of $T$}{
                \lIf{$T' \not\in \mathcal S$}{
                    Add $T'$ to $\mathcal Q$ and to $\mathcal S$
                }
            }
        }
    }
\end{algorithm}

Next, we consider the delay of \Cref{algo:enum:all}.
To this end, we analyze the maximum out-degree of $\mathcal G$. 
If the maximum out-degree is bounded by polynomial in $n$
then we can output all the solutions with polynomial-delay by using a breadth-first search for $\mathcal G$.
The number of type-I neighbors is bounded by $\order{nm\Delta}$ since we have $n$ choices for $e$, $m$ choices for $f$, and $\Delta$ choices for $g$. 
Since the number of paths with length at most two is $\sum_{v \in V}d(v)^2 \le m\Delta$, the number of type-II neighbors is $\order{nm\Delta}$.
Recall that the type-III neighbor of a minimal connected edge dominating set with respect to $e$ is unique. Hence, the number of type-III neighbors is $\order{n}$. 
Thus, the maximum out-degree in $\mathcal G$ is $\order{nm\Delta}$, and we obtain the following theorem.



\begin{theorem}
    There is an $\order{nm^2\Delta}$-delay and exponential space algorithm for enumerating all minimal connected edge dominating sets.
\end{theorem}
\begin{proof}
    The correctness holds from \Cref{lem:edom:scc}.
    We analyze the delay of the algorithm. 
    The bottleneck of \Cref{algo:enum:all} is enumeration of all neighbors and checking for duplicates.
    Since we can compute $\comp{\cdot}$ in $\order{m}$ time by \Cref{lem:comp} and the number of neighbors is $\order{nm\Delta}$, we can enumerate all neighbors in $\order{nm^2\Delta}$ time. 
    By using a standard data structure, i.e., a radix tree, we can check the duplication in $\order{m}$ time and we can insert a solution in $\order{m}$ time. 
    Since the number of neighbors is $\order{nm\Delta}$, duplication checking can be done in $\order{nm^2\Delta}$ time in total.
    Thus, the delay of the algorithm is $\order{nm^2\Delta}$. 
\end{proof}

\section{\texorpdfstring{$K$}{K}-best enumeration of minimal connected edge dominating sets}
We modify \Cref{algo:enum:all} so that the modified algorithm approximately outputs $k$-best minimal connected edge dominating sets. 
To achieve this, we use a priority queue instead of the queue used in \Cref{algo:enum:all}. 
The priority of the elements of the queue is defined as the cardinality of elements. 
Moreover, the priority of elements with the same cardinality is allowed the arbitrary priority.
See \Cref{algo:enum:best} for the details.
To show the correctness of \Cref{algo:enum:best}, we analyze the cardinality of solutions on a path between minimal connected edge dominating sets $X$ and $Y$ in the supergraph $\mathcal G$.
When the cardinality of each minimal connected edge dominating set on the path is sufficiently small, 
we can approximately solve the $k$-best enumeration problem.

\DontPrintSemicolon
\begin{algorithm}[t]
    \caption{A $k$-best enumeration algorithm for minimal connected edge dominating sets in $\order{nm^2\Delta}$ delay with approximation factor $4$. In this algorithm, we allow a polynomial-time preprocessing to compute the first solution. A priority of each element in $Q$ is defined as the cardinality of elements}
    \label{algo:enum:best}
    \Procedure{\best{$G, k$}}{
        Let $T^*$ be a minimal connected edge dominating set of $G$ founded using a $2$-approximation algorithm in \cite{DBLP:journals/ipl/ArkinHH93}\;
        Add $T^*$ to a priority queue $\mathcal Q$ and to set $\mathcal S$\;
        \While{$\mathcal Q$ is not empty $\land$ $k > 0$}{
            Let $T \in Q$ be a minimal connected edge dominating set with the minimum cardinality\;
            Output $T$, decrease $k$ by one, and delete $T$ from $\mathcal Q$\;
            \ForEach{$T'$ is a neighbor of $T$}{
                \lIf{$T' \not\in \mathcal S$}{
                    Add $T'$ to $\mathcal Q$ and to $\mathcal S$
                }
            }
        }
    }
\end{algorithm}

\begin{lemma}\label{lem:edom:path}
    For any pair of minimal edge dominating sets $X$ and $Y$, 
    $\mathcal G$ has a directed path $\pi = (Z_1 = X, \ldots, Z_\ell = Y)$ from $X$ to $Y$ such that any solution on $\pi$ has cardinality at most $\size{X} + 2\size{Y}$.
\end{lemma}
\begin{proof}       
    Let $Z$ be a minimal connected edge dominating set such that $Z \subseteq X \cup Y$. 
    From \Cref{lem:connected}, if $G[Z \cup Y]$ has a cycle $C$, then $Z$ has a neighbor $\tilde{Z}$ that satisfies $\tilde{Z} \subseteq Z \cup Y \subseteq X \cup Y$. 
    
    Suppose that $G[Z \cup Y]$ has no cycles. 
    By \Cref{lem:tree}, we can find a path from $Z$ to $\tilde{Z}$ such that $\tilde{Z} \subseteq X \cup Y$. 
    Let $Z'$ be a type-III neighbor of $Z$ with respect to an edge $e$ and a solution on the path from $Z$ to $\tilde{Z}$.
    According to the definition of type-III neighbors, we first add an edge set $F$ to $Z$. 
    The size of $F$ is at most $\size{Y} + 1$ since $Y$ dominates edges in $F$ and $G[Y]$ is connected. 
    Thus, $\size{Z'} \le \size{Z} + \size{F} - 1 \le \size{Z} + \size{Y}$. 
    Moreover, according to the proof of \Cref{lem:tree}, 
    on a path from $Z'$ to $\tilde{Z}$, 
    any minimal connected edge dominating set $Z''$ on the path satisfies that 
    $Z'' \subseteq (Z \cup Y \cup F) \setminus \set{e} \subseteq (X \cup Y \cup F) \setminus \set{e}$.
    
    From the above discussion, we can decompose $\pi$ to subpaths $\pi_1, \pi_2, \dots$ so that 
    for any $i$, 
    the first and last minimal connected edge dominating sets of $\pi_i$ are subsets of $X \cup Y$ and any intermediate minimal connected edge dominating set of $\pi_i$ has size at most $\size{X \cup Y} + \size{Y}$. 
    Thus, the maximum cardinality of a minimal connected edge dominating set on $\pi$ is at most $\size{X \cup Y} + \size{Y} \le \size{X} + 2\size{Y}$.
\end{proof}

Finally, we show that our algorithm can solve the $k$-best minimal connected edge dominating set enumeration problem.
Let $\mathcal S$ be a set of solutions that are already outputted,
$X$ be the solution in $\mathcal S$,
and $Y$ be a minimum solution not included in $\mathcal S$.
From \Cref{lem:edom:path}, $\mathcal G$ has a directed path connecting $X$ and $Y$ with a path $\pi$ such that solutions corresponding to nodes of $\pi$ have cardinality at most $\size{X} + 2\size{Y}$.
Thus, by traversing $\mathcal G$ using a priority queue,
we show that we only output solutions with cardinality at most $\size{X} + 2\size{Y}$ until $Y$ is outputted.
From this observation, the following theorem can be shown.

\begin{theorem}
    Suppose that there is an $f(n, m)$-time $c$-approximation algorithm for a minimum connected edge dominating set.
    Then, 
    we can approximately enumerate $k$-best minimal connected edge dominating sets in $\order{nm^2\Delta}$ delay with approximation ratio $c+2$ after an $f(n,m)$-time preprocessing.
\end{theorem}
\begin{proof}
    We prove the theorem using induction on $k$.
    When $k$ is equal to one, the statement holds since we have an $f(n, m)$-time $c$-approximation algorithm for the minimum connected edge dominating set problem.
    Thus, we consider the induction step.
    
    Let $\mathcal S$ be a set of minimal connected edge dominating sets that are already outputted, 
    $X$ be a minimal connected edge dominating set in $\mathcal S$ with the minimum cardinality, and
    $Y$ be a minimal connected edge dominating set with the minimum cardinality not contained in $\mathcal S$. 
    We first show the following claim.
    \begin{claim}
        The out-neighbors $N^+_{\mathcal G}(\mathcal S)$ contain a minimal connected edge dominating set $Z$ that satisfies $\size{Z} \le \size{X} + 2\size{Y}$, where $N^+_{\mathcal G}(\mathcal S)$ is the out-neighbors of $\mathcal S$ in $\mathcal G$.
    \end{claim}
    \begin{claimproof}
        From \Cref{lem:edom:scc,lem:edom:path}, $\mathcal G$ has a path $\pi$ from $X$ to $Y$ such that 
        the cardinality of every minimal connected edge dominating set on $\pi$ is at most $\size{X} + 2\size{Y}$.
        Since $\mathcal S$ contains $X$ and does not contain $Y$, $N^+_{\mathcal G}(\mathcal S)$ contains at most one element in $\pi$.
        Since the cardinality of every element on $\pi$ is at most $\size{X} + 2\size{Y}$, the claim holds.
    \end{claimproof}
    
    In \Cref{algo:enum:best}, a minimal connected edge dominating set $X$ is contained in $\mathcal Q$ if and only if $X$ is contained in $N^+_{\mathcal G}(\mathcal S)$.
    By the above claim, the cardinality of a minimum element in $Q$ is at most $\size{X} + 2\size{Y}$.
    Hence, we output a minimal connected edge dominating set with cardinality at most $\size{X} + 2\size{Y} \le (c+2)\size{Y}$ until we output $Y$.

    Finally, we analyze the delay of the algorithm. 
    The algorithm is the same as \Cref{algo:enum:all} except that a queue is changed to a priority queue.
    Moreover, since the number of solutions is at most $2^m$, we can insert a new solution in $\order{m}$ time. 
    Thus, the delay of the algorithm is $\order{nm^2\Delta}$ after $f(n, m)$-time preprocessing.
\end{proof}


Arkin et al.~\cite{DBLP:journals/ipl/ArkinHH93} give a polynomial-time $2$-approximation algorithm for a minimum connected edge dominating set. 
Thus, by using their algorithm as a preprocessing, 
we can immediately obtain an $\order{nm^2\Delta}$-delay approximate enumeration algorithm for $k$-best minimal connected edge dominating sets with an approximation ratio $4$.
Recall that 
the delay of the algorithm is $\order{nm^2\Delta}$. 
However, we allow polynomial-time preprocessing time to find the first solution.


\section*{Acknowledgement}
This work is partially supported by JST CREST Grant Number JPMJCR18K3, JSPS KAKENHI Grant Numbers JP19H01133, JP19K20350, JP20H05793, JP21K17812, JP21H05861, JP22H03549, and JP22H03549, and JST ACT-X Grant Number JPMJAX2105, Japan.

\bibliographystyle{plain}
\bibliography{main.bib}

\begin{thebibliography}{}

\end{thebibliography}


\begin{thebibliography}{10}

\bibitem{DBLP:conf/icalp/AgarwalHS021}
Pankaj~K. Agarwal, Xiao Hu, Stavros Sintos, and Jun Yang.
\newblock Dynamic enumeration of similarity joins.
\newblock In {\em Proc. of {ICALP} 2021}, volume 198 of {\em LIPIcs}, pages
  11:1--11:19. Schloss Dagstuhl - Leibniz-Zentrum f{\"{u}}r Informatik, 2021.

\bibitem{DBLP:conf/icde/AjamiC19}
Zahi Ajami and Sara Cohen.
\newblock Enumerating minimal weight set covers.
\newblock In {\em Proc. of {ICDE} 2019}, pages 518--529. {IEEE}, 2019.

\bibitem{DBLP:journals/ipl/ArkinHH93}
Esther~M. Arkin, Magn{\'{u}}s~M. Halld{\'{o}}rsson, and Refael Hassin.
\newblock Approximating the tree and tour covers of a graph.
\newblock {\em Inf. Process. Lett.}, 47(6):275--282, 1993.

\bibitem{DBLP:conf/ijcai/BasteFJMOPR20}
Julien Baste, Michael~R. Fellows, Lars Jaffke, Tom{\'{a}}s Masar{\'{\i}}k,
  Mateus de~Oliveira~Oliveira, Geevarghese Philip, and Frances~A. Rosamond.
\newblock Diversity of solutions: An exploration through the lens of
  fixed-parameter tractability theory.
\newblock In {\em Proc. of {IJCAI} 2020}, pages 1119--1125. ijcai.org, 2020.

\bibitem{DBLP:journals/talg/BonamyDHPR20}
Marthe Bonamy, Oscar Defrain, Marc Heinrich, Michal Pilipczuk, and
  Jean{-}Florent Raymond.
\newblock Enumerating minimal dominating sets in {$K_t$}-free graphs and
  variants.
\newblock {\em {ACM} Trans. Algorithms}, 16(3):39:1--39:23, 2020.

\bibitem{DBLP:journals/corr/abs-2004-09885}
Yixin Cao.
\newblock {Enumerating Maximal Induced Subgraphs}.
\newblock In {\em Proc. of {ESA} 2023}, volume 274 of {\em LIPIcs}, pages
  31:1--31:13, Dagstuhl, Germany, 2023. Schloss Dagstuhl -- Leibniz-Zentrum
  f{\"u}r Informatik.

\bibitem{Cohen::2008}
Sara Cohen, Benny Kimelfeld, and Yehoshua Sagiv.
\newblock Generating all maximal induced subgraphs for hereditary and
  connected-hereditary graph properties.
\newblock {\em J. Comput. Syst. Sci.}, 74(7):1147--1159, 2008.

\bibitem{DBLP:conf/kdd/ConteMSGMV18}
Alessio Conte, Tiziano~De Matteis, Daniele~De Sensi, Roberto Grossi, Andrea
  Marino, and Luca Versari.
\newblock {D2K:} scalable community detection in massive networks via
  small-diameter k-plexes.
\newblock In {\em Proc. of {KDD} 2018}, pages 1272--1281. {ACM}, 2018.

\bibitem{DBLP:reference/algo/Eppstein16}
David Eppstein.
\newblock \emph{k}-best enumeration.
\newblock In {\em Encyclopedia of Algorithms}, pages 1003--1006. Springer New
  York, 2016.

\bibitem{DBLP:journals/jcss/FaginLN03}
Ronald Fagin, Amnon Lotem, and Moni Naor.
\newblock {Optimal aggregation algorithms for middleware}.
\newblock {\em J. Comput. Syst. Sci.}, 66(4):614--656, 2003.

\bibitem{Gabow:Two:1977}
Harold~N. Gabow.
\newblock {Two Algorithms for Generating Weighted Spanning Trees in Order}.
\newblock {\em SIAM J. Comput.}, 6(1):139--150, 1977.

\bibitem{DBLP:conf/aaai/HanakaKKO21}
Tesshu Hanaka, Yasuaki Kobayashi, Kazuhiro Kurita, and Yota Otachi.
\newblock Finding diverse trees, paths, and more.
\newblock In {\em Proc {AAAI} 2021}, pages 3778--3786. {AAAI} Press, 2021.

\bibitem{DBLP:conf/aaai/HaraI18}
Satoshi Hara and Masakazu Ishihata.
\newblock Approximate and exact enumeration of rule models.
\newblock In Sheila~A. McIlraith and Kilian~Q. Weinberger, editors, {\em Proc.
  of {AAAI} 2018}, pages 3157--3164. {AAAI} Press, 2018.

\bibitem{doi:10.1287/ijoc.2020.1028}
Arne Herzel, Stefan Ruzika, and Clemens Thielen.
\newblock Approximation methods for multiobjective optimization problems: A
  survey.
\newblock {\em INFORMS J Comput}, 2021.

\bibitem{DBLP:conf/fct/KanteLMN11}
Mamadou~Moustapha Kant{\'{e}}, Vincent Limouzy, Arnaud Mary, and Lhouari
  Nourine.
\newblock Enumeration of minimal dominating sets and variants.
\newblock In {\em Proc. of {FCT} 2011}, volume 6914 of {\em LNCS}, pages
  298--309. Springer, 2011.

\bibitem{DBLP:journals/siamdm/KanteLMN14}
Mamadou~Moustapha Kant{\'{e}}, Vincent Limouzy, Arnaud Mary, and Lhouari
  Nourine.
\newblock On the enumeration of minimal dominating sets and related notions.
\newblock {\em {SIAM} J. Discret. Math.}, 28(4):1916--1929, 2014.

\bibitem{Kante:Limouzy:WG:2015}
Mamadou~Moustapha Kant{\'e}, Vincent Limouzy, Arnaud Mary, Lhouari Nourine, and
  Takeaki Uno.
\newblock A polynomial delay algorithm for enumerating minimal dominating sets
  in chordal graphs.
\newblock In {\em Proc. of {WG} 2015}, pages 138--153. Springer, 2015.

\bibitem{DBLP:conf/wads/KanteLMNU15}
Mamadou~Moustapha Kant{\'{e}}, Vincent Limouzy, Arnaud Mary, Lhouari Nourine,
  and Takeaki Uno.
\newblock Polynomial delay algorithm for listing minimal edge dominating sets
  in graphs.
\newblock In {\em Proc. of {WADS} 2015}, volume 9214 of {\em LNCS}, pages
  446--457. Springer, 2015.

\bibitem{Kimelfeld:Efficiently:2008}
Benny Kimelfeld and Yehoshua Sagiv.
\newblock Efficiently enumerating results of keyword search over data graphs.
\newblock {\em Inf. Syst.}, 33(4-5):335--359, 2008.

\bibitem{Kobayashi:Efficient:2020}
Yasuaki Kobayashi, Kazuhiro Kurita, and Kunihiro Wasa.
\newblock Efficient constant-factor approximate enumeration of minimal subsets
  for monotone properties with cardinality constraints.
\newblock {\em CoRR}, abs/2009.08830, 2020.

\bibitem{DBLP:journals/corr/abs-2105-04146}
Yasuaki Kobayashi, Kazuhiro Kurita, and Kunihiro Wasa.
\newblock Polynomial-delay and polynomial-space enumeration of large maximal
  matchings.
\newblock In {\em Proc. of {WG} 2022}, volume 13453 of {\em LNCS}, pages
  342--355, Cham, 2022. Springer International Publishing.

\bibitem{DBLP:journals/corr/abs-2012-09153}
Tuukka Korhonen.
\newblock Listing small minimal separators of a graph.
\newblock {\em CoRR}, abs/2012.09153, 2020.

\bibitem{Lawler1972}
Eugene~L. Lawler.
\newblock {A Procedure for Computing the K Best Solutions to Discrete
  Optimization Problems and Its Application to the Shortest Path Problem}.
\newblock {\em Manage. Sci.}, 18(7):401--405, 1972.

\bibitem{Munaro}
Andrea Munaro.
\newblock {\em On some classical and new hypergraph invariants}.
\newblock PhD thesis, Universit{\'e} Grenoble Alpes, 2016.

\bibitem{Murty:Letter:1968}
Katta~G. Murty.
\newblock {Letter to the Editor―An Algorithm for Ranking all the Assignments
  in Order of Increasing Cost}.
\newblock {\em Oper. Res.}, 16(3):682--687, 1968.

\bibitem{DBLP:conf/pods/RavidMK19}
Noam Ravid, Dori Medini, and Benny Kimelfeld.
\newblock {Ranked Enumeration of Minimal Triangulations}.
\newblock In {\em Proc. of {PODS} 2019}, pages 74--88, 2019.

\bibitem{DBLP:conf/cikm/SadeC20}
Liron Sade and Sara Cohen.
\newblock Diverse enumeration of maximal cliques.
\newblock In {\em Proc. of {CIKM} 2020}, pages 3321--3324. {ACM}, 2020.

\bibitem{DBLP:conf/fimi/UnoKA04}
Takeaki Uno, Masashi Kiyomi, and Hiroki Arimura.
\newblock {LCM} ver. 2: Efficient mining algorithms for frequent/closed/maximal
  itemsets.
\newblock In {\em Proc. of {FIMI} 2004}, volume 126 of {\em {CEUR} Workshop
  Proceedings}. CEUR-WS.org, 2004.

\bibitem{DBLP:conf/sigmod/YangRLG18}
Xiaofeng Yang, Mirek Riedewald, Rundong Li, and Wolfgang Gatterbauer.
\newblock Any-k algorithms for exploratory analysis with conjunctive queries.
\newblock In {\em Proc. of {ExploreDB} 2018}, pages 2:1--2:3. {ACM}, 2018.

\end{thebibliography}

\end{document}